\documentclass{llncs}
\usepackage{llncsdoc}
\usepackage{makeidx}
\usepackage{graphicx}
\graphicspath{{../eps/}{../ps/}}
\usepackage{psfrag}    
\usepackage{url}   
\usepackage{amsmath}
\usepackage{amsfonts}
\usepackage{amssymb}
\newcommand{\Z}{\mathbb{Z}}
\begin{document}
%
% paper title
\title{DNA Image Pro - A Tool for Generating Pixel Patterns using DNA Tile Assembly}
\author{Dixita Limbachiya, Dhaval Trivedi and Manish K Gupta
	\institute{Dhirubhai Ambani Institute of Information and Communication Technology\ Gandhinagar, \ Gujarat-382007, \ India}
	\email{dlimbachiya@acm.org,201001107@daiict.ac.in, mankg@computer.org}
	}
\maketitle

\begin{abstract}
Self-assembly is a process found everywhere in the Nature. In particular, it is known that DNA self-assembly is Turing universal.
Thus one can do arbitrary computations or build nano-structures using DNA self-assembly. In order to understand the DNA 
self-assembly process, many mathematical models have been proposed in the literature. In particular, abstract Tile Assembly Model (aTAM)
received much attention. In this work, we investigate pixel pattern generation using aTAM. For a given image, a tile assembly system is given which can generate the image by self-assembly process. We also consider image blocks with specific cyclic pixel patterns (uniform shift and non uniform shift) self assembly. A software, DNA Image Pro, for generating pixel patterns using DNA tile assembly is also given.

%Using DNA tile assembly, an attempt is made to generate the image by considering each pixel of an image as a single tile. Tile sets for the image can be generated with or without compression. To facilitate the image generation from DNA tile assembly, DNA Image Pro is developed. DNA Image Pro facilitates different color assignment to the tiles which helps in identifying the tile assembly more preciously. DNA Image Pro provides optimization algorithms for image processing with tile assembly. 
\end{abstract}


%%%%%%%%%%%%%%%%%%%%%%%%%%%%%%%%%%%%%%%%%%%%%%%%%%%%%%%%%%%%%%%%%%%%
\section{INTRODUCTION}
 
 Self-assembly of DNA has shown remarkable potential for diverse areas of applications since last two decades. It has been used to build 3D cube \cite{Seeman1982237}, 2D lattice \cite{park2005three}, \cite{zhang2006periodic}, \cite{labean2000construction}, polyhedra
 \cite{he2008hierarchical}, octahedron \cite{he2010chirality}, DNA nanotubes \cite{liu2004DNA}, 2D arrays \cite{he2005self}, complex nano structures \cite{lin2006DNA}. In early days people used trial and error approach to construct the nan-scale structure until E. Winfree shows that self assembly of DNA is Turing universal by proposing a model for algorithmic self-assembly \cite{winfree1998algorithmic}, \cite{fujibayashi2007toward}. Two types of tile assembly models have been studied widely viz. abstract Tile Assembly Model (aTAM) and the kinetic Tile Assembly Model (kTAM) \cite{doty2012theory}. Winfree also developed a simulator Xgrow \cite{Xgrow} to visualize the DNA tile assembly growth with specific tile set as an input file.
 
 In this work, we investigate pixel pattern generation using aTAM. Each pixel of an image can be mapped to unique tile and by obtaining a rule and set of tiles in aTAM one can re-construct the image using DNA self-assembly. Thus, in other words, we obtain an aTAM representation of the image and a .tile file (an input for the Xgrow simulator of Winfree) for the corresponding image. Such a representation could be useful for the image compression. If we want to communicate an image we can send the .tile file which is often less in size than the original image size, however, the computational time cost is more for generating the image. It would be an interesting task to use redundancy in the image and discover minimal sets of tiles for the image generation. Thus in our preliminary study we focus our attention to some simple set of cyclic pixel blocks and study their tile sets.
 
 The paper is organized as follows. Section 2 introduced DNA tiles systems. Section 3 considers the different pixels patterns using DNA tile assembly. Section 4 gives an overview about DNA Image Pro software. An interesting open question is given in Section 5 and the final section discusses about the software availability.
%With the notion of using DNA tile assembly for pixel patterns, method of image block generation using DNA tile assembly is developed here. The pattern of the pixels if understood can be processed by algorithmic tile assembly. DNA Pixel Pro provides a platform to understand the image block generation with prospective of algorithmic tile assembly. This work introduces a new way of generating the tiles assembly using the shift operations.  

.

%\begin{figure}[h]
%\centering
%\includegraphics[width=1.0in]{tile}
%\caption{Example of an tile type}
%\label{fig_sim}
%\end{figure}

%%%%%%%%%%%%%%%%%%%%%%%%%%%%%%%%%%%%%%%%%%%%%%%%%%%%%%%%%%%%%%%%%%%%%%%%%%%
\section{DESIGN OF DNA TILES}
Each tile has label L and  some glue on each side. Each glue has some strength. Two adjacent tiles can join if and only if their glue matches. While assembly, a tile can be transformed but cannot be rotated. A tile system can be represented as $S_R = (T; \textbf{S}; g; \tau )$ where T is the tile system, S is the seed configuration of this tile assembly, g is the glue strength of an edge with configuration \{N E S W\} (N-North, E-East, S-South, W-West) as shown in Fig \ref{DNAtile} and $\tau$  (always $\geq 0)$ is the threshold temperature. A tile file is created using different tile types depending on g and $\tau$.  Xgrow simulator takes a tile file as input and simulates it depending upon the tile types and their glues. In any tile file we need to specify the following basic parameters.

\begin{enumerate}
\item Total number of tile types
\item Total number of glues 
\item Configuration of each tile
\item Seed tile
\item $G_se$ and $G_mc$ parameters that measures free energy and entropy
\end{enumerate}

\begin{figure}[h]
\centering
\includegraphics[width=2.0in]{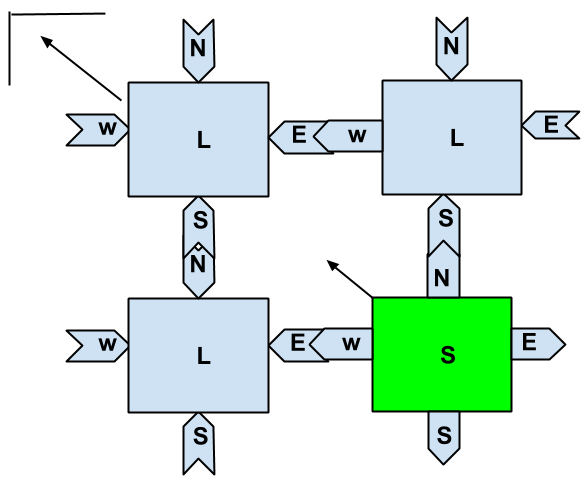}
\caption{Basic DNA tile set design}
\label{DNAtile}
\end{figure}

%\begin{example}
%Tile file \\
%tile edges matches \{\{N E S W\}*\}\newline
%num tile types=7\newline
%num binding types=3\newline
%tile edges=\{\newline
%\{3 0 0 3\}(blue)\newline
%\{3 0 3 1\}(red)\newline
%\{2 3 0 3\}(yellow)\newline
%\{1 1 2 2\}(black)\newline
%\{2 2 1 1\}(black)\newline
%\{1 2 2 1\}(white)\newline
%\{2 1 1 2\}(white)\newline
%\}
%binding strengths=\newline
%\{1 1 2\}\newline
%seed=500,500,1\newline
%Gse=10\newline
%Gmc=19\newline
%block=17\newline
%size=8\newline
%\end{example}

%We can compile this tile file in Xgrow using "\textit{xgrow filename.tiles}". Figure \ref{eg} is the output of given example tile file.

%\begin{figure}[h]
%\centering
%\includegraphics[width=1.6in]{chess}
%\caption{Output of example tile file}
%\label{eg}
%\end{figure}

%%%%%%%%%%%%%%%%%%%%%%%%%%%%%%%%%%%%%%%%%%%%%%%%%%%%%%%%%%%%%%%%%%%%%%%
\section{PIXEL PATTERNS USING DNA TILE ASSEMBLY}
Every image is made up of pixels. To generate image using tile assembly, one can consider each pixel as a tile. Assigning each tile with the pixel color, image block can be easily generated by combining each tile. If we consider image block with dimensions N$\times$M pixels then N$\times$M tile types are needed to generate the image block by tiles assembly. This way is trivial and inefficient for image with large dimensions as it needs large number of tiles. To simplify the problem, consider an image with some fixed pixel patterns. If there is some specific pattern, for which tile sets can be developed such that it can generate complete image, then this can be more feasible and efficient. With this idea, here two basic patterns in image block is considered. If image has cyclic blocks of uniform shift or non uniform shifts, then it can be generated by assembly of few tiles set with uniform and non uniform shift tile sets respectively.

\subsection{Uniform Shift Generator}
Consider image block with uniform shift, then by developing a tile set for self assembly of the uniform shift patterns complete image block can be generated. For uniform shift following basic tile types are required. 

\begin{enumerate}
\item Seed Tile
\item Base Row
\item Base Column
\item Computational Tiles
\end{enumerate}

Depending on uniform shift value, each row will be shifted uniformly. Total number of tiles required for uniform shift generation is given by theorem \ref{uniform}.

\begin{theorem}
Let $\sum_1 = \{i \ | \ 0 \leq i \leq 2N\}$ be the set of glues where $N \in \Z_n$  and let $T_1$ be a set of tiles (different tiles types of $T_1$ as given in Table \ref{uniformtiletable} ) over $\sum_1$ as described in Fig \ref{uniformtiles} %with seed configuration as defined in Fig \ref{uniformtiles} 
then $T_1$ computes the uniform shift $S \in \Z_n$. The total number of tiles required for uniform shift generation is 3N-1.
\label{uniform}
%such that for base row tile ${X = [2,N], Y=X+N+1}$, for base column tile$ {X=[2,N], Y=X+N-1, 
%T=(T-S) \mod N$} and for computational tiles  $X = [1, N], T=(X-1-S) \mod N O/p=(X-S) \mod N$ with g = 1, $\tau = 2$ }
\end{theorem}

\begin{proof}
Consider a tile system $T_1$ over $\sum_1$ with seed tile configuration as shown in Fig \ref{uniformtiles}. Total number of base row, base column and computational tiles gives overall tile sets required. For uniform shift generation, base row tiles are N, base column is N-1 and computational tiles are N. Hence total tile required are N+N-1+N. Therefore minimum number of tiles types needed are 3N-1.
\end{proof}

\begin{figure}[ht]
\centering
\includegraphics[width=5.0in]{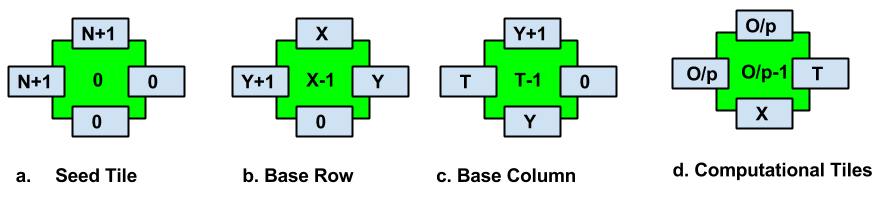}
\caption{Generalized tile system for uniform shift generation. a. Seed tile b. Base row tile where $2 \leq X \leq N$ and Y= X+N-1 c. Base Column tile where $2 \leq X \leq N$, Y= X+N-1 and $T= (T-S) \mod N$ for $T_0$ = 1 d. Computational tile where $1 \leq X \leq N$, $T= (X-1-S) \mod N$ , $O/p=(X-S) \mod N$ }
\label{uniformtiles}
\end{figure}

\begin{table}[ht]
\caption{Number of tiles for uniform shift generation}
\centering
\begin{tabular}{|l|l|}
\hline
Type of tiles&Number of tiles \\ \hline
Seed tiles&1 \\ \hline
Base row tile& N-1 \\ \hline
Base column tile&N-1  \\ \hline
Computational tile&N\\ \hline
Total tiles& 3N-1\\ \hline
\end{tabular}
\label{uniformtiletable}
\end{table}

\begin{figure}[h]
\centering
\includegraphics[width=4.0in]{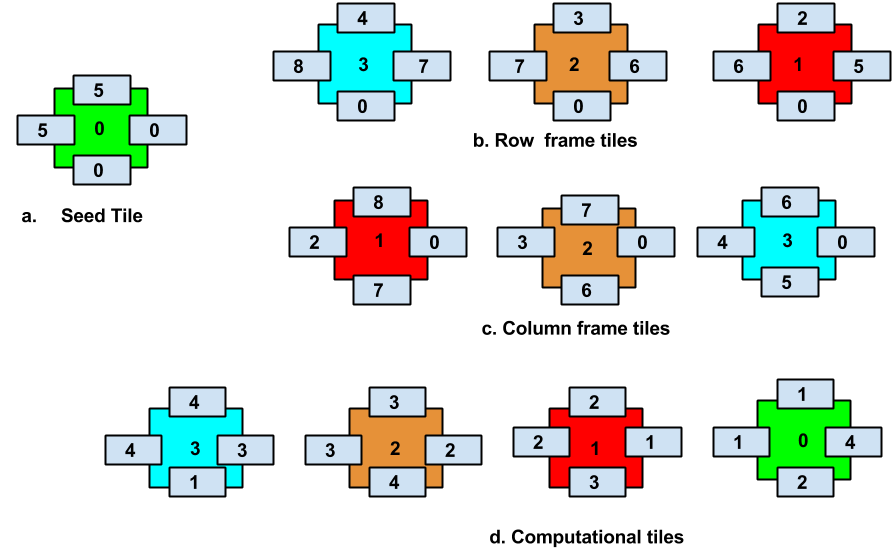}
\caption{Tile sets of uniform tile generation with $N=4$ ans $S=1.$ a. Seed tile b.Base row tile c. Base Column tile d. Computational tiles}
\label{uniformeg}
\end{figure}

\begin{figure}[h]
\centering
\includegraphics[width=3.0in]{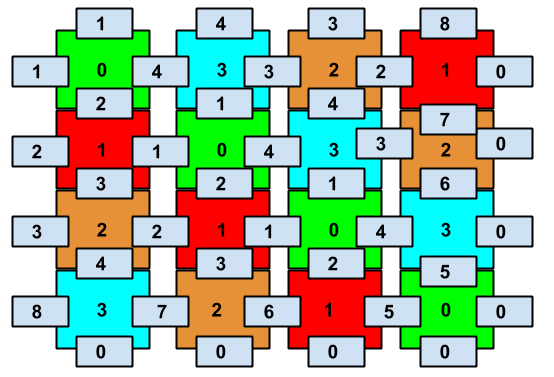}
\caption{Uniform Shift generation with N=4 and S=1 with the given tile system in Fig \ref{uniformeg}. }
\label{uniformgeneration}
\end{figure}

\begin{example}
Let us consider a tile set for $N=4$ and Shift value $S=1.$ Tile set file generated by using DNA Image Pro as shown in Fig \ref{uniformeg} is given below. Tile assembly generated is shown in Fig \ref{uniformgeneration}\\

\%seed tile \newline
\{5 0 0 5 \}(-33554177)\newline
\newline
\%First row. (bottom boundary row)\newline
\newline
\{2 5 0 6 \}(-16842752)\newline
\{3 6 0 7 \}(-33554177)\newline
\{4 7 0 8 \}(-16842752)\newline
\newline
\%First column. (Right most column)\newline
\newline
\{6 0 5 4 \}(-16842752)\newline
\{7 0 6 3 \}(-33554177)\newline
\{8 0 7 2 \}(-16842752)\newline
\newline
\%Rule tiles\newline
\newline
\{1 4 2 1 \}(-33554177)\newline
\{2 1 3 2 \}(-16842752)\newline
\{3 2 4 3 \}(-33554177)\newline
\{4 3 1 4 \}(-16842752)\newline\newline
\label{eg1}
\end{example}

In the example \ref{eg1}, the seed tile is $\{5,0,0,5\}$ with configuration $\{N+1,0,0,N+1\}$ for $N=4.$ For base row $\{X,Y,0,Y+1\}$ where $2 \leq X \leq 4$ and $Y = X+N-1.$ \newline
For $X=2,$ tile is $\{2,5,0,6\}$ \newline
\-\ \hspace{8mm} $X=3$ tile is $\{3,6,0,7\}$ \newline
\-\ \hspace{8mm} X=4 tile is \{4,7,0,8\}\newline
\\
For Base Column is with configuration \{Y+1,0,Y,T\}, Y= X+N-1 where $2 \leq X \leq 4$, $T = (T-S) \mod N$ and T = T+N if T $<$ 1.\newline
First we will have T=1\\
Now for X=2, Y=5, T=4 which will be \{6,0,5,4\}.\\
\-\ \hspace{6.5mm}for X=3, Y=6, T=3 which will be \{7,0,6,3\}\\
\-\ \hspace{6.5mm}for X=4, Y=7, T=2 which will be \{8,0,7,2\}\\\\
%This is also same as we had in the output file.\\\\
For Computational tiles we have configuration \{Op,T,X,Op\}, for $1 \leq X \leq 4$, $T= (X-1-S)\mod N$ and T = T+N if T$<$1, 
$Op = (X-S)\mod N$  and Op = Op+N if Op $<$1.\newline\\
Let T=1.\\
For X =1, T = 3, Op = 4 which will be \{4,3,1,4\}.\\
for X =2, T = 4, Op = 1 which will be \{1,4,2,1\}.\\
for X =3, T = 1, Op = 2 which will be \{2,1,3,2\}.\\
for X =4, T = 2, Op = 3 which will be \{3,2,4,3\}.\\

\subsection{Non-Uniform Shift Generator General tile set}
Consider block of non uniform shifts in image block, generation of tile assembly for non uniform shift image block can be generated.For each row different shifts values can be assigned. For Given base row of length N and Non-Uniform Shift array S, We can represent the tile set using following type of tiles:  \begin{enumerate}
\item Seed Tile
\item Base Row
\item Base Column
\item Computational Tiles
\end{enumerate}
Generalized tiles for each type are N = Base Row Length and S = Array of Shift Value.

\begin{theorem}
let $\sum_2 = \{i \ | \ 0 \leq i \leq 2N\}$ be the set of glues where $N \in \Z_n$ and let $T_2$ be a set of tiles (different tiles types of $T_2$ are given in Table \ref{nonuniformtiletable} ) over $\sum_2$ as described in Fig \ref{nonuniformtiles} then $T_2$ computes the non uniform shift $S \in \Z_n$. The total number of tiles required for uniform shift generation is 2(N-1)+ N$\times$S.

%such that for base row tile ${X = [2,N], Y=X+N+1}$, for base column tile$ {X=[2,N], Y=X+N-1, 
%T=(T-S) \mod N$} and for computational tiles  $X = [1, N], T=(X-1-S) \mod N O/p=(X-S) \mod N$,$1 \leq L \leq N-1$ }
\end{theorem}

\begin{proof}
Consider a tile system $T_2$ over $\sum_2$ with seed tile configuration as shown in Fig \ref{nonuniformtiles}. For non uniform shift, each N-1 base row tiles and N-1 column tiles are shifted non uniformly with computational tiles N $\times$ S, therefore total tiles are $2(N-1) + N \times S$.
\end{proof}

\begin{figure}[ht]
\centering
\includegraphics[width=5.0in]{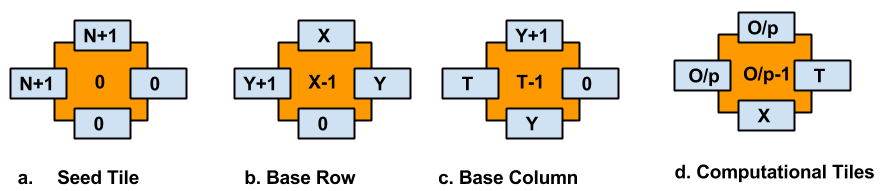}
\caption{Generalized tile system for non uniform shift generation. a. Seed tile b. Base row tile where $2 \leq X \leq N$ and Y= X+N-1 c. Base Column tile where $2 \leq X \leq N$, Y= X+N-1 and $T = (T-S[X-1]) \mod N$ for $T_0$ = 1 d. Computational tile where $1 \leq X \leq N$, $T= (X-1-S) \mod N$ , $O/p=(X-S[L]]) \mod N$ }
\label{nonuniformtiles}
\end{figure}

\begin{table}[ht]
\caption{Number of tiles for non uniform shift generation}
\centering
\begin{tabular}{|l|l|}
\hline
Type of Tiles&Number of tiles \\ \hline
Seed tiles&1 \\ \hline
Base row tile& N-1 \\ \hline
Base column tile&N-1  \\ \hline
Computational tile&N $\times$ S\\ \hline
Total tiles& 2(N-1)+ N$\times$S\\ \hline
\end{tabular}
\label{nonuniformtiletable}
\end{table}

\begin{figure}[h]
\centering
\includegraphics[width=5.0in]{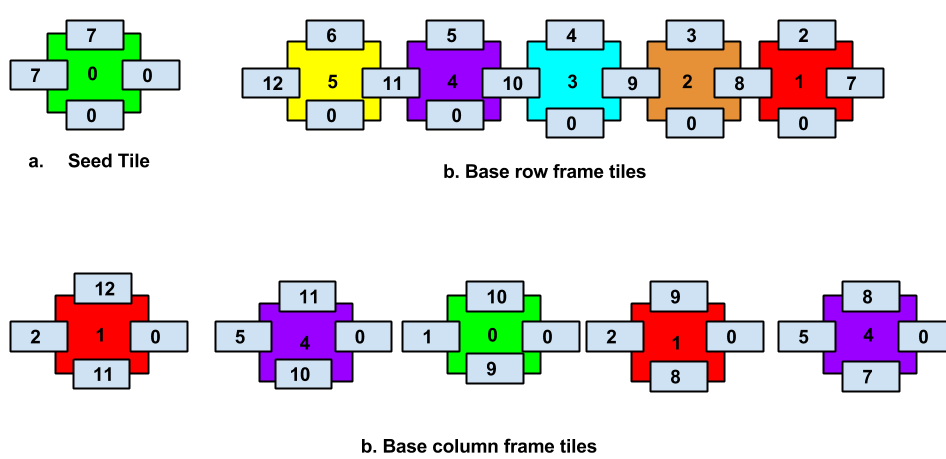}
\caption{Tile sets of non uniform tile generation with N=6 ans S= [2,3,1,2,3]. a. Seed tile b. Base row tile c. Base Column tile} \label{nonuniformeg}
\end{figure}

\begin{figure}[h]
\centering
\includegraphics[width=4.0in]{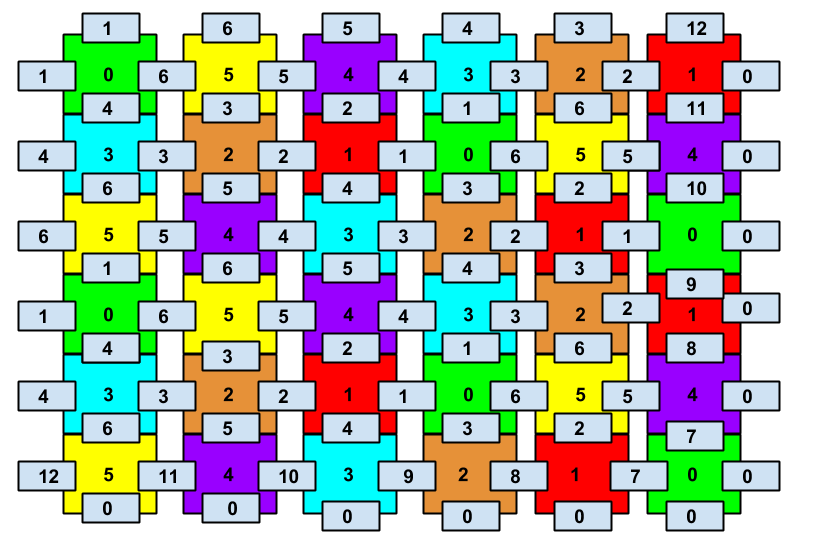}
\caption{Non uniform Shift generation with N=6 and S= [2,3,1,2,3] with the given tile system in Fig \ref{nonuniformeg}.}
\label{nonuniformgeneration}
\end{figure}

\begin{example}
Let us consider a tile set for N=6 and non uniform Shift value S= [2,3,1,2,3]. Tile set file generated by using DNA Image Pro as shown in Fig \ref{nonuniformeg}. Tile assembly generated for non uniform shift is shown in Fig \ref{nonuniformgeneration}.\\
\end{example}

%%%%%%%%%%%%%%%%%%%%%%%%%%%%%%%%%%%%%%%%%%%%%%%%%%%%%%%%%
\subsection{Row Transformation}
Row transformation on pixel configuration is considered in this section. Transformation of the base row tiles will effect the shift generation of the tile assembly. Let the image with P configuration be transformed to image with Q configuration by some transformation T. This will result in the tile assembly transformation. Lemma \ref{rowtranform} gives the method to apply the row transformation of the DNA tiles using shift operator.

\begin{lemma}
Row transformation of a DNA tile set of image block can be obtained by ${\sum_{j=1}^i S[j] \hspace{1mm}  mod  \hspace{1mm} N}$ where S is the shift operator and N is number of tiles in the base row.
\label{rowtranform}
\end{lemma}

%As we know a tile file in our examples can be described as \\
%1) Base Row  color array\\
%2) Shift Array\\\\
%Now to get the row transformation we need to modify both the array accordingly. \\
%Let’s say our original color array is C[] and shift array is S[]. \\
%Accumulator[i] = $\sum\limits_{j=1}^i S[j]$\\
%New Shift[i] = (Accumulator[i])\%N\\
%To get the shift for base row we will rotate color array by S[0].\\
%New Shift and color array will generate the row transformation of input file.\\

%%%%%%%%%%%%%%%%%%%%%%%%%%%%%%%%%%%%%%%%%%%%%%%%%%%%%%%%%%
\subsection{Image Tile generator}
%If the image dimensions are N$\times$M then N*M tile types are needed to generate the image by tiles assembly. 
For any input image block, one can generate the tile assembly. By rendering the image selected, the tile sets are generated. There are two methods to generate the image.
\begin{enumerate}
\item Rapid image generator
\item Normal image generator
\end{enumerate}

Unlike normal image generation, rapid image generator will compress the original image in 32$\times$32 and generate the tile file for compressed image. Normal image generator will not compress the input image so it will use more tile types then rapid image generator. Image will be generated at high speed by using rapid generator but it will decrease the quality of image. Figure \ref{smiley} is the original image which is to be generated by using rapid image generation. Figure \ref{rapidimage} shows the image generated by tile assembly using rapid image generation.

\begin{figure}[h]
\centering
\includegraphics[width=1.5in]{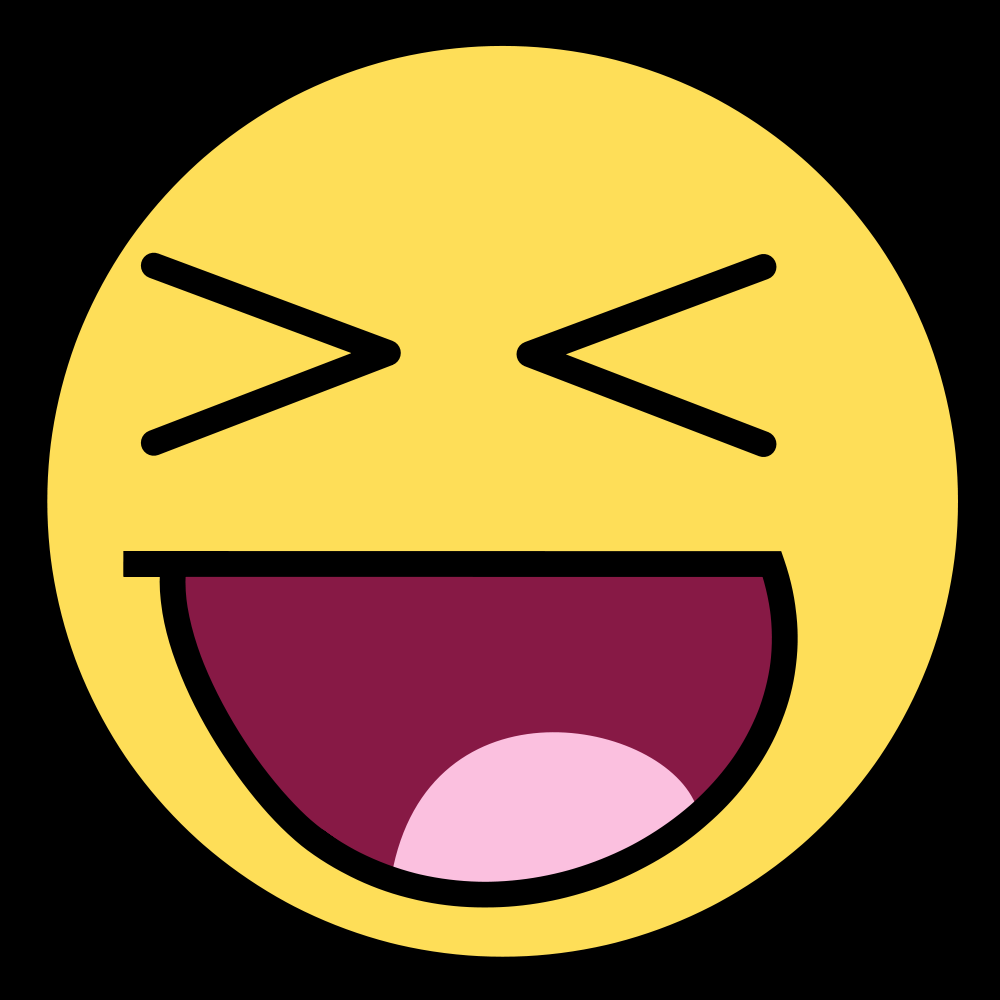}
\caption{Original image to be generated using tiles assembly}
\label{smiley}
\end{figure}

\begin{figure}[ht]
\centering
\includegraphics[width=1.5in]{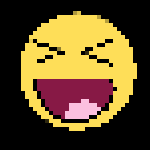}
\caption{Image generated by rapid image generator in Xgrow}
\label{rapidimage}
\end{figure}

To estimate the tile files required for particular image, Fig \ref{image_tilesize} shows the relation between image size and size of tiles required. As the image size increases, number of tiles required increases. 

\begin{figure}
\centering
\includegraphics[scale=0.9]{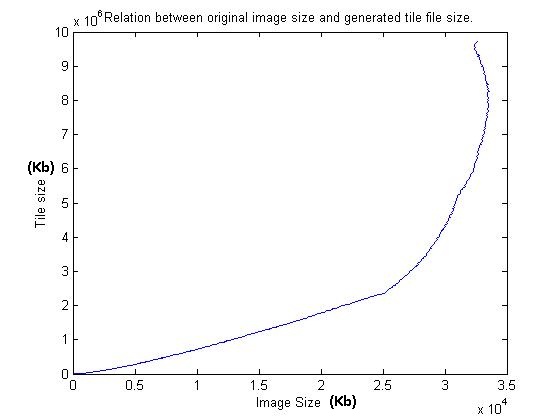}
\caption{Plot shows the relation between tile file size required to generate image.}
\label{image_tilesize}
\end{figure}

To estimate the time required to generate particular image by tile assembly, Fig \ref{sizetime} shows the relation between image size and time required to generate the image. As the image dimension increases, time required for tile assembly increases. 

\begin{figure}[h]
\centering
\includegraphics[scale=0.5]{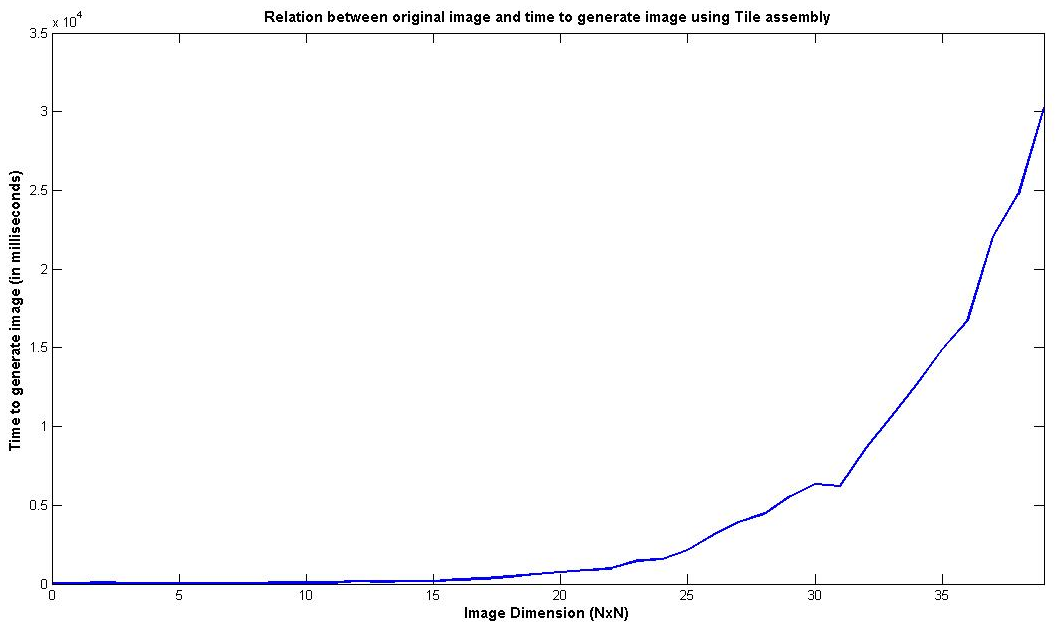}
\caption{Plot shows the relation between time required to generate the image.}
\label{sizetime}
\end{figure}

%%%%%%%%%%%%%%%%%%%%%%%%%%%%%%%%%%%%%%%%%%%%%%%%%%%%%%%%%%%%%%%%%%%%%%%%%%%%%%
\section{DNA Image Pro: GUI OVERVIEW}
DNA Image Pro is used to generate tile file for different types of images and it also provides new version of Xgrow.
It has menu with the following option:
\begin{itemize}
\item Uniform Shift Generator
\item Non Uniform Shift Generator
\item Transformation Generator
\item Image Tile generator
\item Run a Tile File
\end{itemize}

%%%%%%%%%%%%%%%%%%%%%%%%%%%%%%%%%%%%%%%%%%%%%%%%%%%%%%%%%%%%%%%%%%%%%%%%%%%%%%
\subsection{Uniform Shift Generator}
This will generate the N$\times$N tile sets on the basis of the seed tile sets specified. User should specify the flakes tile set that is the base row for the tile assembly. It will generate the tile assembly uniformly with given value of the shift operator. It will have uniform pattern for the tile assembly. Figure \ref{baserowuniform} shows the image tile generated by DNA Image Pro with number of tiles N = 6 and uniform shift operator S = 2. Figure \ref{uniformshift2} is the result of tile set simulation in Xgrow for the example N = 6 and uniform shift operator S = 2.

\begin{figure}[h]
\centering
\includegraphics[width=3.0in]{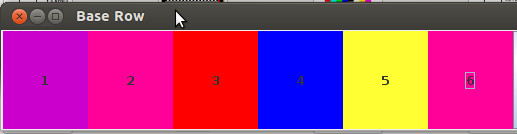}
\caption{Base row for uniform shift generator}
\label{baserowuniform}
\end{figure}

\begin{figure}[h]
\centering
\includegraphics[width=2.0in]{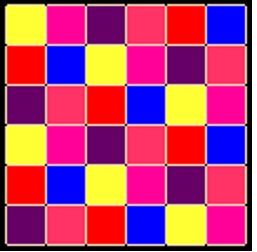}
\caption{Xgrow simulation of tile set with N=6 and S = 2 by uniform shift generator}
\label{uniformshift2}
\end{figure}

%%%%%%%%%%%%%%%%%%%%%%%%%%%%%%%%%%%%%%%%%%%%%%%%%%%%%%%%%%%%%%%%%%%%%
\subsection{Non Uniform Shift Generator}
It is similar to uniform generator. This will generate the N$\times$N tile sets on the basis of the seed tile sets specified. User should specify the flakes tile set that is the base row for the tile assembly.  It will generate the tile assembly non uniformly with given value of the shift operators. Unlike uniform generator, there will be shift operator for each row. So the number of shift operator will be N-1. Figure \ref{basenonuniform} shows the image tile generated by DNA Image Pro with number of tiles N = 6 and non uniform shift operator S = [1,2,3,1,2]. Figure \ref{nonuniformshift} is the result of tile set simulation in Xgrow for the example N = 6 and uniform shift operator S = [1,2,3,1,2].

\begin{figure}[h]
\centering
\includegraphics[width=3.0in]{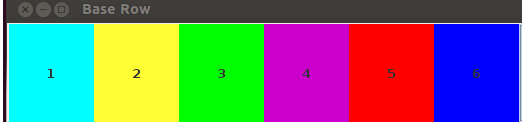}
\caption{Base row for non uniform shift generator}
\label{basenonuniform}
\end{figure}

\begin{figure}[h]
\centering
\includegraphics[width=2.0in]{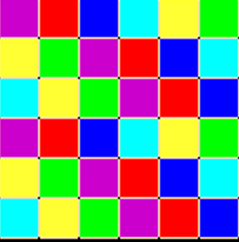}
\caption{Xgrow simulation of tile set with N = 6 and shift S= [1,2,3,1,2] by non uniform shift generator}
\label{nonuniformshift}
\end{figure}

\subsection{Transformation Generator}
 It is similar to non-uniform generator along with the row transformation on the base tiles. This will generate the N$\times$N tile sets on the basis of the seed tile sets specified. It will generate the tile assembly by applying the row transformation on the base tile set using the shuffle value provided.
 
%%%%%%%%%%%%%%%%%%%%%%%%%%%%%%%%%%%%%%%%%%%%%%%%%%%%%%%%%%%%%%%%%%%%%%%%
\subsection{Run a Tile File}
This option is used to execute a tile file generated by DNA Image Pro in Xgrow. User need to select a file and it will execute it in Xgrow.

%\subsection{Image Generation }
%Every pixel of an image can be considered as a single tile. So we can generate an Image of N$\times$M dimensions using N$\times$M tiles and appropriate glues.

%\section{Lower bounds and results}
%\begin{lemma}Minimum number of tile types needed for uniform shift generation are 3N - 1. \end{lemma}
%\begin{proof}
%Consider number of tiles N for uniform shift generation. Total number of tiles required be sum total of tiles for base row , base column and rule tiles. For uniform shift generation, base row tiles are N, base column is N-1 and rule tiles are N. Hence total tile required are N+N-1+N. Therefore minimum number of tiles types needed are 3N-1.
%\end{proof}

%\begin{lemma} Minimum number of tile types needed for non uniform transformation are N$\times$ Number of unique Shift + 2N -1.\\
%This is a loose upper bound.  \end{lemma} %For some cases it will be less than this.
%\begin{proof}

%%%%%%%%%%%%%%%%%%%%%%%%%%%%%%%%%%%%%%%%%%%%%%%%%%%%%%%%%%%%%%%%%%%%%%%%%%%
\section{OPEN PROBLEM}
We are considering each pixel of an image as unique tile type. So if we want to generate an image with width W and length L, we would require W$\times$L tile types. For any given image, optimization of the tile set assembly is still an open challenge. One can use the redundancy in the image to search for the minimal set of tiles needed to generate the image.

%%%%%%%%%%%%%%%%%%%%%%%%%%%%%%%%%%%%%%%%%%%%%%%%%%%%%%%%%%%%%%%%%%%%%%%%%
\section{Software Availability}
The software (source code), installers, user manual, installation guide and other related materials can be downloaded from \\ 
http://www.guptalab.org/dnaimagepro/.

%%%%%%%%%%%%%%%%%%%%%%%%%%%%%%%%%%%%%%%%%%%%%%%%%%%%%%%%%%%%%%%%%%%%%%
\section*{Acknowledgments}\label{sec:Acknowledgments}
Authors would like to thank Paul W. K. Rothemund for useful comments on the paper.

\bibliographystyle{splncs03}
\bibliography{dnaimage}

\end{document}